\documentclass[11pt, letterpaper]{article}
\usepackage[margin=1in]{geometry}

\usepackage{amsfonts, amsmath, amssymb, amsthm}
\usepackage{bm}
\usepackage{verbatim}
\usepackage{color}
\usepackage{euscript}
\usepackage{graphicx}
\usepackage[usenames,dvipsnames]{xcolor}
\usepackage{paralist}

\usepackage{enumitem}
\usepackage{stmaryrd}

\usepackage{xspace}
\usepackage{wrapfig}

\def\vgap{\vspace{0.1 in}}

\newtheorem{theorem}{Theorem}
\newtheorem{lemma}{Lemma}

\newtheorem{observation}{Observation}

\newtheorem{remark}{Remark}

\def\S{\mathcal{S}}
\def\Q{\mathcal{Q}}
\def\rep{\mathit{rep}}
\def\capp{\mathit{capp}}
\def\eapp{\mathit{\vare app}}

\newcommand{\IR}{\mathbb{R}}
\newcommand{\Oe}{O_{\vare}}

\newcommand{\vare}{\varepsilon}

\newcommand{\define}{\textit}

\newcommand{\IntRange}[1]{\left\llbracket #1 \right\rrbracket}

\begin{document}

\title{Approximate Range Counting Revisited\footnote{This research was partly supported by a Doctoral Dissertation Fellowship (DDF) from the Graduate School of University of Minnesota.}}

\date{}

\author{Saladi Rahul\\
Department of Computer Science and Engineering \\
University of Minnesota \\
{\em sala0198@umn.edu}}

\maketitle

\begin{abstract}
We study range-searching for colored objects, where one has to count
(approximately) the number of colors present in a query range.
The problems studied mostly involve orthogonal range-searching in two
and three dimensions, and the dual setting of rectangle stabbing by
points.
We present optimal and near-optimal solutions for these problems.
Most of the results are obtained via reductions to the approximate
uncolored version, and improved data-structures for them.
An additional contribution of this work is the introduction of nested
shallow cuttings.
\end{abstract}

\section{Introduction}
\begin{wrapfigure}{r}{0.25\textwidth}
\includegraphics[scale=1]{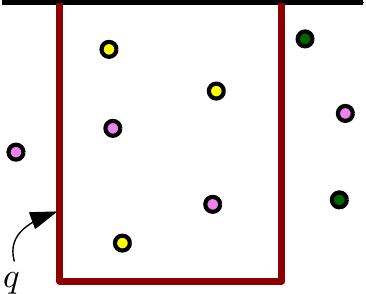}
\label{fig:colored-problems}
\end{wrapfigure}
 Let $S$ be a set of $n$ 
geometric objects in $\IR^d$ which are segregated into disjoint groups (i.e., {\em colors}). 
Given a query $q\subseteq \IR^d$, a color $c$  {\em intersects (or is present in)} $q$ if any object
in $S$ of color $c$ intersects $q$, and let $k$  be the number of colors of $S$ present in $q$. In the 
{\em approximate colored range-counting problem}, the task is to preprocess $S$ into a  data structure, so that 
for a query $q$, one can efficiently report the {\em approximate} number of colors present in $q$. 
Specifically, return any value  in the range \\ $[(1-\vare)k,(1+\vare)k]$, where $\vare \in (0,1)$ 
is a pre-specified parameter.

 Colored range searching and its related problems have 
been  studied before 
\cite{kn11,nv13,ptsnv14,m02,lps08,ksv06,gjs04,bkmt95,agm02,gjs95,jl93,krsv07,lp12,lw13,n14,sj05}. 
They  are known as GROUP-BY queries 
in the database literature. A popular variant is the \define{colored orthogonal range searching} problem: 
$S$ is a set of $n$ colored points in $\IR^d$, and $q$ is an axes-parallel rectangle. 
As a motivating example for this problem, consider the following query:
  ``How many countries have employees aged between $X_1$ and $X_2$ while earning annually more than 
  $Y$ rupees?". An employee is represented as a colored point $(age, salary)$, 
  where the color encodes the country, and the query is the axes-parallel 
  rectangle $[X_1,X_2] \times [Y,\infty)$.


\subsection{Previous work and background}

In the {\em standard} approximate range counting problem 
there are no colors. One is  interested in the approximate number of 
objects intersecting the query. Specifically, 
if $k$ is the number of objects of $S$ intersecting $q$,
then return a value  in the range 
$[(1-\vare)k,(1+\vare)k]$.

\paragraph{$\vare$-approximations.} 
In the {\em additive-error $\vare$-approximation}, a set $Z \subseteq S$ is picked such that, 
given a query $q$, we {\em only} inspect $Z$ and return a value which lies in the range 
$[k-\vare n,k+\vare n]$. 
Vapnik and Chervonenkis~\cite{vc71} proved that 
a random sample $Z$ 
 of size $O(\frac{\delta}{\vare^2}\log\frac{\delta}{\vare})$ provides an 
$\vare$-approximation with good probability, where 
 $\delta$ is the  VC-dimension ($\delta$ is usually a constant).

\paragraph{Relative $(p,\vare)$-approximation.} 
 Har-Peled and Sharir~\cite{hs11}
 introduced the notion of {\em relative $(p,\vare)$-approximation} for geometric settings.
The goal is to pick a {\em small} set $Z \subset S$ which can be used to compute a relative approximation for queries with large value of $k$. Formally, 
given a parameter $p\in (0,1)$, a set $Z \subset S$ is a relative $(p,\vare)$-approximation if:
\begin{align*}
|Z\cap q|\cdot \frac{n}{|Z|} &\in
\begin{cases}
 [(1-\vare)k,(1+\vare)k]  \quad \quad \text{if } k\geq pn \\
 [k-\vare pn,k+\vare pn]  \quad \quad \text{otherwise. }
\end{cases}
\end{align*}
Har-Peled and Sharir
prove that a sample $Z$ from $S$ of size $O\left(\frac{1}{\vare^2 p}\left(\delta\log\frac{1}{p} + \log\frac{1}{q} \right)\right)$ will succeed with probability at least $1-q$.

Har-Peled and Sharir construct relative $(p,\vare)$-approximations for point sets 
and halfspaces in $\IR^d$, for $d\geq 2$, 
and use them to answer approximate counting  for any query which contains more than $pn$ points. A nice feature of these results is that they 
are {\em sensitive} to the value of $k$. Specifically, the larger the value of $k$ is, the faster the query is answered. The intuition is that the larger 
the value of $k$ is, the larger is the error the query is allowed to make and hence, a smaller sample suffices. 
Even though relative $(p,\vare)$-approximations give a relative approximation 
only for queries with large values of $k$, Aronov and Sharir~\cite{as10}, and Sharir and Shaul~\cite{ss11} incorporated them into 
data structures which give an approximate count for all values of $k$.

\paragraph{General reduction to companion problems.}
Aronov and Har-Peled~\cite{ah08}, and Kaplan, Ramos and Sharir~\cite{krs11} 
presented general techniques to answer  approximate range counting queries. 
In both instances, the authors reduce the task of answering an
approximate counting query, into answering a few queries in
data-structures solving an easier {\em (companion)} problem.
Aronov and Har-Peled's companion problem is the emptiness query, 
where the goal is to report whether $|S\cap q|=0$.
Specifically, assume that there is  a data structure of size $S(n)$ which  answers the emptiness query 
in $O(Q(n))$ time. Aronov and Har-Peled show that there is a data structure of size 
$O(S(n)\log n)$ which answers the approximate counting query in $O(Q(n)\log n)$ time (for simplicity we 
ignore the dependency on $\vare$). 
Kaplan {\em et al.}'s companion problem is the  
 range-minimum query, where each object of $S$ has a weight associated with it and the goal is to report the 
object in $S\cap q$ with the minimum weight.

Even though the reductions of \cite{ah08} and \cite{krs11} 
seem different, there is an interesting discussion in Section~$6$ of \cite{ah08} about the underlying ``sameness" of both 
techniques.

\paragraph{Levels.} 
 Informally, for a set $S$ of $n$ objects, a {\em $t$-level} of $S$ is a 	surface 
 such that if a point $q$ lies 
above (resp., on/below) the surface, then the number of objects of 
$S$ containing $q$ is $>t$ (resp., $\leq t$).
Range counting can be reduced in some
cases to deciding the level of a query point. Unfortunately, the
complexity of a single level is not well understood. For example, for hyperplanes in the
plane, the $t$-level has super-linear complexity
$\Omega(n 2^{\sqrt{\log t}})$~\cite{t00} in the worst-case (the known upper bound is
$O(n t^{1/3})$~\cite{d98} and closing the gap is a major open problem). 
In particular, the prohibitive
complexity of such levels makes them inapplicable for the approximate
range counting problem, where one shoots for linear (or near-linear)
space data-structures. 

\paragraph{Shallow cuttings.}  
A \emph{$t$-level shallow cutting} is a set of simple cells, that lies
strictly below the $2t$-level, and their union  covers all
the points below (and on) the $t$-level. For many geometric objects in two and three dimensions,
such $t$-shallow cuttings have $O(n/t)$ cells~\cite{aes99}. Using such cuttings
leads to efficient data-structures for approximate range counting. Specifically, one uses binary search on a ``ladder'' of
approximate levels (realized via shallow cuttings) to find the
approximation.

For halfspaces in $\IR^3$, Afshani and Chan~\cite{ac09} avoid doing the binary search and 
find the two consecutive levels in  optimal $O(\log \frac{n}{k})$ expected time.
Later, Afshani, Hamilton and Zeh \cite{ahz10} obtained a worst-case optimal 
solution for many geometric settings. 
Interestingly, their results hold in the pointer machine model, the I/O-model and the cache-oblivious model.
However, in the word-RAM model their solution is not optimal and the query time is 
$\Omega(\log\log U + (\log\log n)^2)$. 

\vgap
\noindent
{\bf Specific problems.} Approximate counting for orthogonal range searching in $\IR^2$ was studied by Nekrich \cite{n14}, and Chan and Wilkinson \cite{cw13} 
in the word-RAM model. In this setting, the input set is points in $\IR^2$ and the query is a rectangle in $\IR^2$. 
A hyper-rectangle in $\IR^d$ is  {\em $(d+k)$-sided} 
if it is bounded on both sides in $k$ out of the $d$ dimensions  and 
unbounded on one side in the remaining $d-k$ dimensions.
Nekrich \cite{n14} presented a data structure for approximate colored 
$3$-sided range searching in $\IR^2$, where the input is points 
and the query is a $3$-sided rectangle in $\IR^2$.
However, it has an approximation 
factor of $(4+\vare)$, whereas we are interested in obtaining a 
tighter approximation factor of $(1+\vare)$. To the best of our 
knowledge, this is the only work directly addressing  an approximate 
colored counting query.


\subsection{Motivation}

\paragraph{Avoiding expensive counting structures.} 
A search problem is 
decomposable if given two disjoint sets of objects $S_1$ 
and $S_2$, the answer to $F(S_1\cup S_2)$ can be computed 
in constant time, given the answers to $F(S_1)$ and $F(S_2)$ 
separately.
This property is widely used in the literature~\cite{ae98} for counting 
in standard problems (going back to the work of Bentley and Saxe~\cite{bs80} in the late 1970s). 
For colored counting problems, however, $F(\cdot)$ is not 
decomposable. If $F(S_1)$ (resp. $F(S_2)$) has $n_1$ (resp. $n_2$) 
colors, then this information is insufficient to compute 
$F(S_1 \cup S_2)$, as they might have  common  colors.

As a result, for {\it many exact} colored counting queries the known space and query time bounds are  expensive.
For example, for colored orthogonal  range searching problem in $\IR^d$,  existing structures use $O(n^d)$ space to achieve polylogarithmic query time~\cite{krsv07}. 
Any substantial improvement in the preprocessing time {\em and}  the query time  would lead to a substantial improvement in the best exponent of matrix multiplication \cite{krsv07} 
(which is a major open problem). Similarly,  counting structures  
for colored halfspace counting in $\IR^2$ and $\IR^3$~\cite{gjs04}
are expensive.

Instead of an exact count, if one is willing to settle for an approximate count, then this work presents a data structure
with  $O(n \text{ polylog } n)$ space  and  $O(\text{polylog } n)$ query time.

\paragraph{Approximate counting in the speed of emptiness.} In an emptiness query, 
the goal is to decide if $S\cap q$ is empty. The approximate counting query is at least as hard as 
the emptiness query: When $k=0$ and $k=1$, no error is tolerated. 
Therefore, a natural goal while answering approximate range counting queries is to 
match the bounds of its corresponding {\em emptiness query}.

\subsection{Our results and techniques}

\subsubsection{Specific problems}

The  focus of the paper is building data structures for  approximate colored counting queries, 
 which exactly match or {\em almost} match the 
bounds of their corresponding emptiness problem.

\paragraph{$3$-sided rectangle stabbing in 2d and related problems.} 
In the colored interval stabbing problem,  the input is $n$ colored intervals  with endpoints in $\IntRange{U} = \{1,\ldots, U\}$,
and the query is a point in $\IntRange{U}$. We present a linear-space data structure which answers the approximate counting query 
in $O(\log\log U)$ time.
The new data structure can be used to handle some geometric settings in 2d: 
the  {\em colored dominance search} (the input is a set of $n$ points, and the query is a $2$-sided rectangle) and 
the {\em colored $3$-sided rectangle  stabbing} (the input is a set of $n$ 
$3$-sided rectangles, and the query is a point). The results are summarized in Table~\ref{table:results}. 

\paragraph{Range searching in $\IR^2$.}  
The input is a set of $n$ colored points in the plane.
For  $3$-sided query rectangles, an {\em optimal} solution (in terms of $n$) for approximate counting is obtained.
For $4$-sided query rectangles,  an {\em almost-optimal} solution for approximate counting is obtained. 
The size of our data structure is off by a factor of $\log\log n$ w.r.t. its corresponding emptiness structure 
which occupies 
 $O(n\frac{\log n}{\log\log n})$ space and answers the emptiness query 
   in $O(\log n)$ time~\cite{c86}. The results are summarized in Table~\ref{table:results}.

\paragraph{Dominance search in $\IR^3$.} The input  is a set of $n$ colored points in $\IR^3$ and the query is a 
$3$-sided rectangle in $\IR^3$ (i.e., an octant). An almost-optimal solution is obtained 
requiring  $O(n\log\log n)$ space and $O(\log n)$ time to answer the approximate counting query. 

\begin{table}[t]
\begin{tabular}{|c|c|c|c|c|c|}
\hline
Dime- & Input, & New Results& Previous Approx.& Exact Counting & Model\\
-nsion & Query &  &  Counting Results& Results&\\
\hline\hline
$1$ & intervals, & S: $n,$ & S: $n,$& S: $n,$&\\
 & point & Q: $\log\log U$ & Q: $\log\log U + $& Q: $\log\log U+ \log_wn$& WR\\
 \cline{1-2}
$2$ & points, &  & $(\log\log n)^2$& &\\
& $2$-sided rectangle & && &\\
\cline{1-2}
$2$ & $3$-sided rectangles, & Theorem~\ref{thm:many-colored}    & Remark~\ref{rem:ac-many} & Remark~\ref{rem:ec-many}&\\ 
& point &  & & &\\
\hline\hline
& & & & &\\
$2$ & points, & S: $n,$ & S: $n\log^2n,$ &  &\\
    & $3$-sided rectangle & Q: $\log n$ & Q: $\log^2n$& not studied& PM\\
    &                     & Theorem~\ref{thm::three-sided-color}(A) & Remark~\ref{rem:ac-3-sided}& &\\ 
\hline
& & & & &\\
$2$ & points, & S: $n\log n,$ & S: $n\log^3n,$& S: $n^2\log^6n,$ &\\
    & $4$-sided rectangle & Q: $\log n$ & Q: $\log^2n$& Q: $\log^7n$ & PM\\
    & &  Theorem~\ref{thm::three-sided-color}(B) & Remark~\ref{rem:ac-3-sided} & Kaplan {\em et al.}~\cite{krsv07} &\\
\hline\hline
& & & & &\\
$3$ & points, & S: $n\log^{*}n,$ & S: $n\log^2n,$&  &\\
    & $3$-sided rectangle  & Q: $\log n\cdot\log\log n$ & Q: $\log^2n$& not studied& PM\\
       & & Theorem~\ref{thm:3d-dom} & Remark~\ref{rem:ac-3d-dom} &  &\\
\hline
\end{tabular}
\caption{A summary of the results obtained for several approximate colored counting queries.
To avoid clutter, the $O(\cdot)$ symbol and the dependency on $\vare$ is not shown in 
 the space and the query time bounds.
For the second column in the table, the first entry is the input and the second entry is 
the query. 
For each of the results column in the table, the first entry is the space occupied by the data structure and the second 
entry is the time taken to answer the query. WR denotes the word-RAM model and PM denotes 
the pointer machine model.}
\label{table:results}
\end{table}

\vgap
For the sake of completeness, in Section~\ref{sec:appl-sec-red} we present results for a 
couple of other colored problems which have  expensive exact counting structures.

\subsubsection{General reductions}

We present two general reductions for solving approximate colored counting queries
by reducing them to ``easy" companion queries (Theorem~\ref{thm::main-1} and Theorem~\ref{thm:accq}).

\vgap
\noindent
{\bf Reduction-I (Reporting + $C$-approximation).}
 In the first reduction a colored approximate counting query is answered 
 using two companion structures: (a) {\em reporting structure} (its objective 
is to report the $k$ colors), and  (b) {\em $C$-approximation structure} (its objective is to report any value $z$ s.t. 
$k \in [z,Cz]$, where $C$ is a constant).
Significantly, unlike previous reductions~\cite{ah08,krs11},
there is {\em no asymptotic loss} of efficiency in space and query time bounds w.r.t. to 
the two companion problems. 

\vgap
\noindent
{\bf Reduction-II (Only Reporting).} The second reduction 
 is a modification of the Aronov and Har-Peled~\cite{ah08} reduction. 
We present the reduction  for the following reasons: 
 \begin{compactenum}[a)]
    \item Unlike reduction-I, this reduction is ``easier" to use 
    since it uses only the reporting structure and avoids  
     the $C$-approximation structure. 
    \item The analysis of Aronov and Har-Peled is slightly complicated because of their insistence 
on querying emptiness structures. We show that by using reporting structures the analysis 
becomes simpler. This reduction is useful when the  reporting query is not significantly 
costlier than the emptiness query.
\end{compactenum}

\subsubsection{Our techniques}
The results are obtained  via a non-trivial combination of several techniques. 
For example, 
(a) new reductions from colored problems to standard problems, 
(b) obtaining a linear-space data structure by performing  random sampling on a super-linear-size data structure, 
(c) refinement of path-range trees of Nekrich~\cite{n14} to obtain an optimal data structure for 
$C$-approximation of colored $3$-sided range search in $\IR^2$, and
(d) {\em random sampling on colors} to obtain the two general reductions.

In addition, we introduce  {\em nested shallow cuttings} 
for $3$-sided rectangles in 2d.
The idea of using a hierarchy of cuttings (or samples) is, of course,
not new. However, for this specific setting, we get a hierarchy where
there is no penalty for the different levels being compatible with
each other. Usually, cells in the lower levels have to be clipped to
cells in the higher levels of the hierarchy, leading to a degradation
in performance. In our case, however, cells of the lower levels are
fully contained in the cells of the level above it.

\paragraph{Paper organization.} In Section~$2$, we  present a solution to the 
colored $3$-sided rectangle stabbing in 2d problem. 
In Section~$3$ we present a solution to the colored dominance search in $\IR^3$ problem. 
In Section~$4$ and $5$, the two general reductions are presented.
In Section~$6$, the application of the first reduction to colored orthogonal range search in $\IR^2$ problem is 
shown. In Section~$7$, applications of the second reduction is shown.
Finally, we conclude in Section~$8$.

\section{$3$-sided Rectangle Stabbing in 2d}\label{sec:many-colored}

The goal of this section is to prove the following theorem.
\begin{theorem}\label{thm:many-colored}
Consider the following three colored geometric settings:
\begin{enumerate}
\item {\bf Colored interval stabbing in 1d}, where the  input is a set $S$ of $n$ colored intervals in one-dimension
and the query $q$ is a point. The endpoints of the intervals and the query point lie on a grid $\IntRange{U}$. 
\item {\bf Colored dominance search in 2d},  where the input is a set $S$ of $n$ colored points in 2d and the query $q$ is a quadrant 
of the form $[q_x,\infty) \times [q_y,\infty)$. 
The input points and the point  $(q_x,q_y)$ lie on a grid $\IntRange{U} \times \IntRange{U}$. 
\item {\bf Colored $3$-sided rectangle  stabbing in 2d}, where the  input is a set $S$ of $n$ colored 
$3$-sided rectangles in 2d and the query $q$ is a point. The endpoints of the rectangles 
and the query point lie on a grid $\IntRange{U} \times \IntRange{U}$.
\end{enumerate}
Then there exists an $\Oe(n)$ size word-RAM data structure which can answer an approximate counting query 
for these three  settings in $\Oe(\log\log U)$ time. The notation $\Oe(\cdot)$ hides the dependency on $\vare$.
\end{theorem}

Our strategy for proving this theorem is the following: 
In Subsection~\ref{subsec:trans}, we present a transformation of these three colored 
 problems to the {\em standard} $3$-sided rectangle stabbing in 2d problem. 
Then in Subsection~\ref{subsec:standard}, we construct nested shallow cuttings and use them 
to solve the standard $3$-sided rectangle stabbing in 2d problem.

\subsection{Transformation to a standard problem}\label{subsec:trans}
From now on the focus will be on
colored $3$-sided rectangle stabbing in 2d problem, 
since the geometric setting of (1) and 
(2) in Theorem~\ref{thm:many-colored} are its special cases. 
We present a transformation of the  colored $3$-sided rectangle 
stabbing in 2d problem to the  {\em standard} $3$-sided rectangle 
stabbing in 2d problem. 


Let $S_c \subseteq S$ be the set of $3$-sided rectangles of a color $c$. In the preprocessing phase, we perform the 
following steps: (1) Construct a union of the rectangles of $S_c$. Call it ${\cal U}(S_c)$. 
(2) The vertices of ${\cal U}(S_c)$ include original vertices of $S_c$ and some new vertices.
Perform a {\it vertical decomposition} of ${\cal U}(S_c)$ by shooting a vertical ray 
upwards from every {\em new} vertex of ${\cal U}(S_c)$ till it hits $+\infty$. This leads to a 
decomposition of ${\cal U}(S_c)$ into $\Theta(|S_c|)$ pairwise-disjoint $3$-sided rectangles. 
Call these new set of rectangles ${\cal N}(S_c)$.  

\begin{figure}[h]
 \centering
 \includegraphics[scale=1]{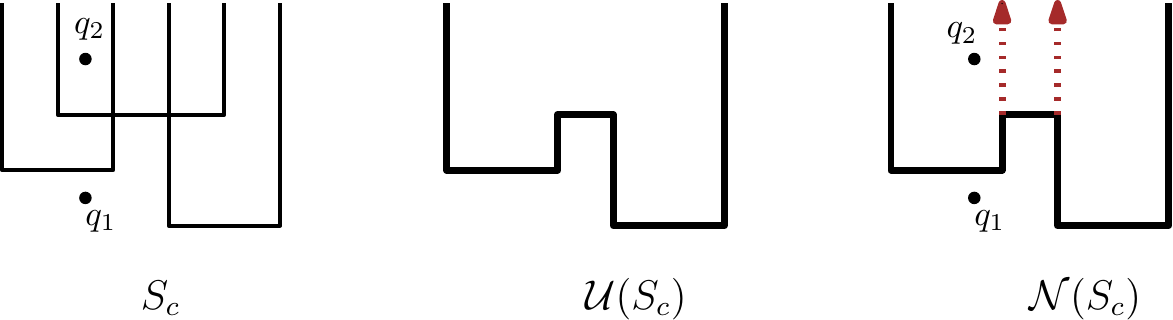}
  \label{fig:trans-I}
 \end{figure}

Given a query point $q$, 
we can make the  following two  observations: 
\begin{itemize}
\item If $S_c \cap q = \emptyset$, then ${\cal N}(S_c) \cap q = \emptyset$. 
See query point $q_1$ in the above figure.

\item  If $S_c \cap q \neq \emptyset$, then exactly one rectangle in ${\cal N}(S_c)$ is stabbed by $q$.
See query point $q_2$ in the above figure.
\end{itemize}

Let ${\cal N}(S)=\bigcup_{\forall c} {\cal N}(S_c)$, and  clearly, $|{\cal N}(S)|=O(n)$. 
Therefore, the colored $3$-sided rectangle stabbing in 2d problem  on $S$ has been reduced 
to   the {\it standard} $3$-sided rectangle stabbing in 2d problem on ${\cal N}(S)$.

\subsection{Standard $3$-sided rectangle stabbing in 2d}\label{subsec:standard}
In this subsection we will prove the  following lemma.
\begin{lemma}\label{lemma:rect-stab}
{\bf (Standard $3$-sided rectangle stabbing in 2d.)} In this geometric setting, the input is 
a set $S$ of $n$ uncolored $3$-sided rectangles of the form $[x_1,x_2] \times [y,\infty)$,
and the query $q$ is a point. The endpoints of the rectangles lie on a grid $\IntRange{U} \times \IntRange{U}$.
There exists a data structure of size $\Oe(n)$ which can answer 
an approximate counting query in $\Oe(\log\log U)$ time.
\end{lemma}

By a standard rank-space reduction, 
 the rectangles of $S$  can be projected to a $\IntRange{2n} \times \IntRange{n}$ grid: 
 Let  $S_x$ (resp., $S_y$) be the list of  the $2n$ vertical (resp., $n$ horizontal) sides  
 of $S$ in increasing order of their $x-$ (resp., $y-$) coordinate value. 
Then each rectangle $r=[x_1,x_2] \times [y,\infty) \in S$ is projected to a rectangle 
 $[rank(x_1), rank(x_2)] \times [rank(y),\infty)$, where $rank(x_i)$ (resp., $rank(y)$) 
 is the index of $x_i$ (resp., $y$) in the list $S_x$ (resp., $S_y$). 
 Given a query point $q \in \IntRange{U} \times \IntRange{U}$, we can 
use the van Emde Boas structure~\cite{b77} to  perform a predecessor search on $S_x$ and $S_y$ in $O(\log\log U)$ time 
to find the position of $q$ on the $\IntRange{2n} \times \IntRange{n}$ grid. 
Now we will focus on the new setting 
and prove the following result.

\begin{lemma}\label{lemma:rect-stab-n}
For the standard $3$-sided rectangle stabbing in 2d problem, 
consider a setting where the rectangles   
have endpoints lying on a grid $\IntRange{2n} \times \IntRange{n}$.
Then there exists a data structure of size $\Oe(n)$ which can answer 
the approximate counting query in $\Oe(1)$ time.
\end{lemma}

\subsubsection{Nested shallow cuttings}

To prove Lemma~\ref{lemma:rect-stab-n},  we will first construct shallow cuttings for $3$-sided rectangles in 2d.
Unlike the general class of shallow cuttings, the shallow cuttings we construct
for $3$-sided rectangles will have a stronger property of cells in the lower level 
lying completely inside the cells of a higher level.

 \begin{lemma}
Let $S$ be a set of $3$-sided rectangles (of the form $[x_1,x_2] \times [y,\infty)$) 
whose endpoints lie on a $\IntRange{2n} \times \IntRange{n}$  grid. 
A {\em $t$-level shallow cutting} of $S$ produces a set ${\cal C}$ of 
 interior-disjoint $3$-sided rectangles/cells of the form $[x_1,x_2] \times (-\infty,y]$.
There exists a set ${\cal C}$ 
with the following three properties:
\begin{enumerate}
\item 
$|{\cal C}| = 2n/t$.
\item If $q$ does not lie inside any of  the cell in ${\cal C}$, then 
$|S\cap q| \geq t$.
\item Each cell in ${\cal C}$ intersects at most $2t$ rectangles of $S$.
\end{enumerate}
\end{lemma}
\begin{proof}
Partition the plane into $\frac{2n}{t}$ vertical slabs, 
such that  $t$ vertical lines of $S$ lie in each slab, i.e., 
each slab has a width of $t$. See Figure~\ref{fig:three-sided-rectangle}(a).
Consider a slab $s=[x_1,x_2] \times (-\infty, +\infty)$. 
Among all the rectangles of $S$ which completely span 
the slab $s$, let $y_t$ be the $y$-coordinate of the rectangle with the $t$-th smallest 
$y$-coordinate. If less than $t$ segments of $S$ span slab $s$, 
then set $y_t:= +\infty$. Let the {\em upper segment} of the slab $s$  be the horizontal 
segment $[x_1,x_2] \times [y_t]$.
Each slab contributes a cell $[x_1,x_2] \times (-\infty, y_t]$ 
to set ${\cal C}$. See Figure~\ref{fig:three-sided-rectangle}(a).

Property~$1$ is easy to verify, since $\frac{2n}{t}$ slabs are 
constructed. To prove Property~$2$, consider a point $q$ 
which lies in slab $s$ but does not lie in the cell 
$[x_1,x_2] \times (-\infty, y_t]$.
This implies that there are at least $t$ 
rectangles of $S$ which contain $q$, and  hence,
 $|S\cap q| \geq t$. To prove Property~$3$, consider a cell 
$r$ and its corresponding slab $s$. 
The rectangles of $S$ which intersect $r$ either span the slab $s$ or 
partially span the slab $s$. By our construction, there can be at most 
$t$ rectangles of $S$ of each type.
\end{proof}

\begin{observation}\label{lemma:containment}
{\em (Nested Property)} Let $t$ and $i$ be integers. Consider a $t$-level and a 
$2^it$-level shallow cutting. By our construction, each cell in $2^it$-level contains 
exactly $2^i$ cells of the $t$-level. More importantly, each cell in the $t$-level 
is contained inside a single cell of $2^it$-level (see Figure~\ref{fig:three-sided-rectangle}(a)).
\end{observation}

\begin{figure}
 \centering
 \includegraphics[scale=1]{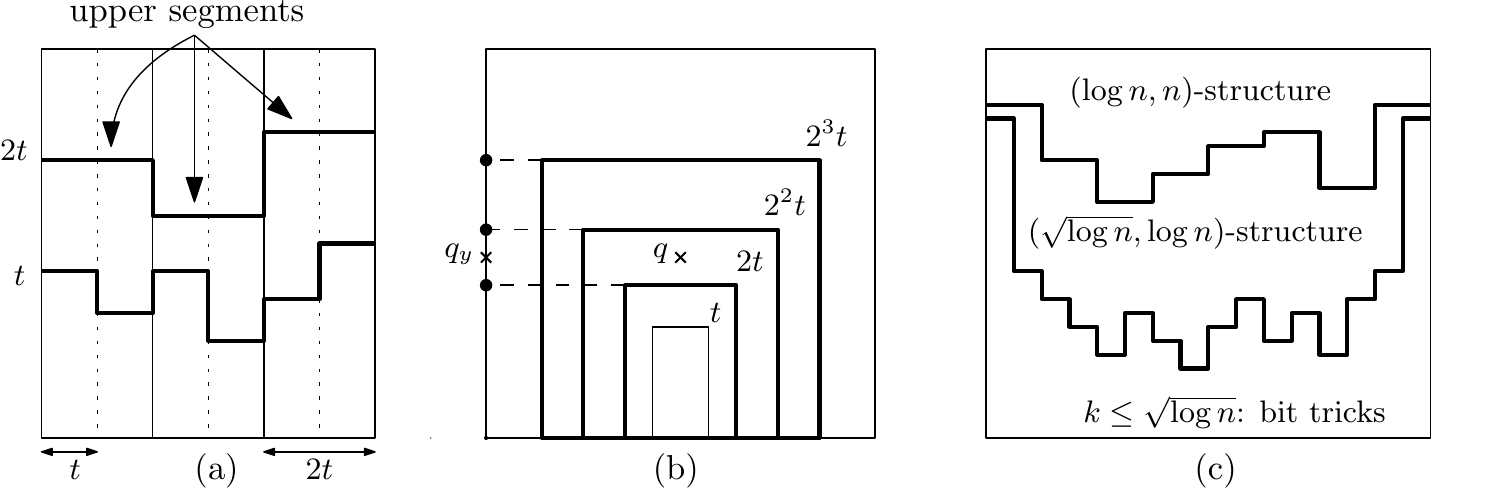}
  \caption{(a) A portion of the $t$-level and $2t$-level is shown. Notice that by our construction, 
  each cell in the $t$-level is contained inside a cell in the $2t$-level. 
  (b) A cell in the $t$-level and the set ${\cal C}_r$ associated with it. (c) 
  A high-level summary of our data structure.}
 \label{fig:three-sided-rectangle}
 \end{figure}

\subsubsection{Data structure}

Now we will use nested shallow cuttings to find a constant-factor approximation for the 
$3$-sided rectangle stabbing in 2d problem. 
In \cite{ahz10}, the authors show how to convert a constant-factor approximation into a 
$(1+\vare)$-approximation for this geometric setting.
The  solution is based on {\em $(t,t')$-level-structure} and {\em $(\leq\sqrt{\log n})$-level shared table}. 

\paragraph{$(t,t')$-level structure.}
Let  $i, t$ and $t'$ be integers s.t. $t'=2^it$. If $q(q_x,q_y)$ lies between the $t$-level and the $t'$-level cutting of $S$, 
then a $(t,t')$-level-structure will answer the approximate 
counting query in $O(1)$ time and occupy
 $O\left(n+\frac{n}{t}\log t'\right)$ space. 

\vgap
{\em Structure.} 
Construct  a shallow cutting of $S$ for levels 
$2^jt, \forall j\in [0,i]$.
For each cell, say $r$, in the $t$-level 
we do the following: Let ${\cal C}_r$ be the set of cells  from 
the $2^1t, 2^2t,2^3t,\ldots,2^it$-level, which 
contain $r$ (Observation~\ref{lemma:containment} guarantees this property). 
Now project the upper segment of each cell of ${\cal C}_r$ onto the $y$-axis (each segment projects to a point). 
Based on the $y$-coordinates of these $|{\cal C}_r|$ projected points build a fusion-tree~\cite{fw93}. 
Since there are $O(n/t)$ cells in the $t$-level and $|{\cal C}_r|=O(\log t')$, the total space occupied  is $O(\frac{n}{t}\log t')$.
See Figure~\ref{fig:three-sided-rectangle}(b).

\vgap
{\em Query algorithm.} 
Since $q_x\in\IntRange{2n}$, it takes $O(1)$ time to find the cell $r$ 
of the $t$-level whose $x$-range contains $q_x$. 
If the predecessor of $q_y$ in ${\cal C}_r$ 
belongs to the $2^jt$-level, then $2^jt$ is a constant-factor approximation of $k$.   
The predecessor query  also takes $O(1)$ time.  

\paragraph{$(\leq\sqrt{\log n})$-level shared table.}
Suppose $q$ lies in a cell in the $\sqrt{\log n}$-level 
shallow cutting of $S$. Then constructing the $(\leq \sqrt{\log n})$-level shared table
will answer the exact counting query in $O(1)$ time. 
We will need the following lemma.

\begin{lemma}
For a cell $c$ in the $\sqrt{\log n}$-level shallow cutting of $S$, its {\em conflict list} $S_c$ 
is the set of rectangles of $S$ intersecting $c$. Although the number of cells in the $\sqrt{\log n}$-level 
is $O\left(\frac{n}{\sqrt{\log n}}\right)$,
the number of combinatorially ``different" conflict lists is merely $O(\sqrt{n})$. 
\end{lemma}
\begin{proof}
Consider any set $S_c$ from the shallow cutting.
By a standard  rank-space reduction the endpoints of $S_c$ will   
lie on a  $\IntRange{2|S_c|} \times \IntRange{|S_c|}$  grid.
Any set $S_c$ on the $\IntRange{2|S_c|} \times \IntRange{|S_c|}$ 
grid can be {\em uniquely} represented using $O(|S_c|\log|S_c|)=O(\sqrt{\log n}\log\log n)$ bits 
as follows:  (a) assign a {\em label} to each rectangle, 
and (b) write down the label of each rectangle in increasing order 
of their $y$-coordinates. The label for a rectangle $[x_1,x_2] \times [y,\infty)$ 
will be $``x_1 x_2"$ which requires $O(\log\log n)$ bits. 
The number of combinatorially different conflict lists
which can be represented using $O(\sqrt{\log n}\log\log n)$ bits is 
bounded by $2^{O(\sqrt{\log n}\log\log n)}=O(n^{\delta})$, for an arbitrarily small $\delta<1$. 
We set $\delta=1/2$.
\end{proof}

{\em Shared table.} Construct a $\sqrt{\log n}$-level 
shallow cutting of $S$. For each cell $c$, perform a rank-space reduction of its conflict list 
$S_c$. Collect the combinatorially different conflict lists.
On each conflict list, the number of combinatorially different queries will be only 
$O(|S_c|^2)=O(\log n)$. 
In a lookup table, for each pair of $(S_c,q)$ we store the exact value of $|S_c \cap q|$.
The total number of entries in the lookup table is $O(n^{1/2}\log n)$.

{\em Query algorithm.} Given a query $q(q_x,q_y)$,  the following three $O(1)$ time operations are performed:
(a) Find the cell $c$ in the $\sqrt{\log n}$-level which contains $q$. 
If no such cell is found, then stop the query and conclude that $k\geq \sqrt{\log n}$.
(b) Otherwise, perform a rank-space reduction on $q_x$ and $q_y$ to map it to the 
$\IntRange{2|S_c|} \times \IntRange{|S_c|}$ grid. Since, $|S_c|=O(\sqrt{\log n})$, 
we can build  fusion trees~\cite{fw93} on $S_c$ to perform the rank-space reduction in $O(1)$ time.  
(c) Finally, search for $(S_c,q)$ in the lookup table and report the exact count.

\paragraph{Final structure.} At first thought, one might be tempted to construct a 
$(0,n)$-level-structure. However, that would occupy $O(n\log n)$ space. 
The issue is that the $(t,t')$-level structure requires super-linear space for small values of $t$. 
Luckily,  the $(\leq\sqrt{\log n})$-level shared table will efficiently handle the small values of $t$.

Therefore, the strategy is to construct the following: 
(a) a $(\leq\sqrt{\log n})$-level shared table,
(b) a $(\sqrt{\log n},\log n)$-level-structure, and (c) a $(\log n,n)$-level-structure.
Now, the space occupied by all the three structures will be $O(n)$. 
See Figure~\ref{fig:three-sided-rectangle}(c) for a summary of our data structure.

\begin{remark}\label{rem:ac-many}
\normalfont For the standard $3$-sided rectangle stabbing in 2d problem, 
a simple binary search on the levels leads to a linear-space data structure with a query time of $\Oe(\log\log U + (\log\log n)^2)$. 
The technique of Afshani~{\em et al.}~\cite{ahz10}
can be used to answer this approximate counting query. However, their analysis works well for structures with query time 
of the form $\log n$ or $\log_Bn$, but breaks down  for structures with query time of the form $\log\log n$.
\end{remark}

\begin{remark}\label{rem:ec-many}
\normalfont If we want an exact count for  the standard $3$-sided rectangle stabbing in 2d problem, 
then the problem can be reduced to exact counting for standard dominance search in 2d~\cite{eo82}.
Jaja~{\em et al.}~\cite{jms04} present a linear-space structure which can answer the exact counting 
for dominance search in 2d in $\Oe(\log\log U + \log_wn )$ time.
\end{remark}

\section{Colored Dominance Search in $\IR^3$}\label{sec:dominance}

\begin{theorem}\label{thm:3d-dom}
In the colored dominance search in $\IR^3$ problem, the input set $S$ is  $n$ 
colored points in $\IR^3$ and the query $q$ is a point.
Then there is a pointer machine data structure of size $\Oe(n\log^{*}n)$ which can answer an approximate 
colored counting query in $\Oe(\log n\cdot \log\log n)$ time.
The notation $\Oe(\cdot)$ hides the dependency on $\vare$.
\end{theorem}

The strategy to prove this theorem is the following. 
First, we reduce the colored dominance search in $\IR^3$ problem to 
a {\em standard} problem of $5$-sided rectangle stabbing in $\IR^3$. 
Then in the remaining section we solve the 
standard $5$-sided rectangle stabbing in $\IR^3$ problem. 

\subsection{Reduction to $5$-sided rectangle stabbing in $\IR^3$}\label{subsec:red-colored-2}
In this subsection we present a reduction of colored dominance search in $\IR^3$ problem to the 
standard $5$-sided 
 rectangle stabbing  in $\IR^3$ problem. 
Let $S$ be a set of $n$ colored points lying in $\IR^3$. Let $S_c \subseteq S$ be the set of points of  color $c$, 
and  $p_1,p_2,\ldots,p_t$ be the  points of $S_c$ in decreasing order of their 
$z$-coordinate value. With each point $p_i(p_{ix},p_{iy},p_{iz})$, we associate a region $\phi_i$ in $\IR^3$ 
which satisfies the following invariant: a point $(x,y,z)$ belongs to $\phi_i$ if and only if 
in the region $[x,+\infty) \times [y,+\infty) \times [z,+\infty)$ the point of $S_c$ with the 
largest $z$-coordinate is $p_i$. The following assignment of regions ensures the 
invariant:
\begin{itemize}
\item $\phi_1 = (-\infty, p_{1x}] \times (-\infty, p_{1y}] \times (-\infty, p_{1z}]$
\item $\phi_i = (-\infty, p_{ix}] \times (-\infty, p_{iy}] \times (-\infty, p_{iz}] \setminus
\bigcup_{j=1}^{i-1} \phi_j, \forall i\in [2,|S_c|]$.
\end{itemize}

By our construction, each region $\phi_i$ is unbounded in the negative $z$-direction.
Each region $\phi_i$ is broken into disjoint $5$-sided rectangles via 
{\it vertical decomposition} in the $xy$-plane (see Figure~\ref{fig:5-sided}). The vertical decomposition 
ensures that the total number of disjoint rectangles generated is bounded by 
$O(|S_c|)$. Now we can observe that (i) if a color $c$ has at least one point inside $q$, 
then exactly one of its transformed 
rectangle will contain $q$, and (ii) if a color $c$ has no point inside $q$, then none of its transformed 
rectangles will contain $q$. Therefore, the colored dominance search in $\IR^3$ has been transformed 
to   the standard $5$-sided rectangle stabbing query.

\begin{figure}[h]
 \centering
\includegraphics[scale=0.7]{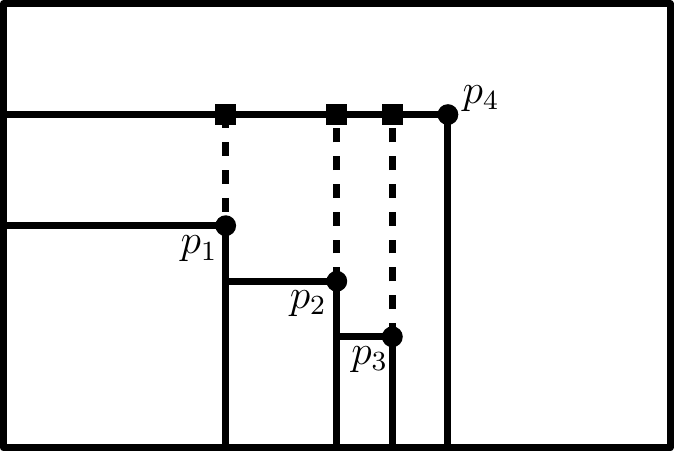}
\caption{A dataset containing four points. 
The projection of $\phi_1,\phi_2,\phi_3$ and 
$\phi_4$ onto the $xy$-plane is shown.
The dashed lines are created during the vertical decomposition. Each 
rectangle created during the vertical decomposition is lifted back 
to  a $5$-sided rectangle in $\IR^3$. }
\label{fig:5-sided}
\end{figure}

\subsection{Initial strcuture}

\begin{lemma}\label{lemma:subopt}
In the standard $5$-sided rectangle stabbing in $\IR^3$ problem, the input is a set $S$ 
 of $n$ $5$-sided rectangles in $\IR^3$  and the query $q$ is a point. 
Then there exists a pointer machine data structure of size $\Oe(n\log\log n)$ 
which can answer an approximate counting query in $\Oe(\log n\cdot \log\log n)$ time. 
\end{lemma}

\noindent
The rest of the subsection is devoted to proving this lemma.

\vgap
\noindent
{\bf Recursion tree.} Define a parameter $t=\log_{1+\vare}n$. We will assume that the $5$-sided rectangles are unbounded 
along the $z$-axis. Consider the projection of the rectangles of $S$ on to the $xy$-plane and impose  an orthogonal 
$\IntRange{2\sqrt{\frac{n}{t}}} \times \IntRange{2\sqrt{\frac{n}{t}}}$ grid such that each horizontal and vertical slab contains the projections of $\sqrt{nt}$ sides of $S$.
 Call this the root of the recursion 
tree. Next, for each vertical and horizontal slab, we recurse on the rectangles of $S$ which are {\em sent} to that slab. 
At each node of the recursion tree, if we have $m$ rectangles in the subproblem,
then $t$ is changed to $\log_{1+\vare}m$ and the grid size changes to  
$\IntRange{2\sqrt{\frac{m}{t}}} \times \IntRange{2\sqrt{\frac{m}{t}}}$.
We stop the recursion when a node has less than $c$ rectangles, for a suitably large constant $c$. 

\vspace{0.1 in}
\noindent
{\bf Assignment of rectangles.} For a node in the tree, 
the intersection of every pair of  horizontal and 
vertical grid line defines a {\em grid point}. 
Each rectangle of $S$ is assigned to $\Oe(\log\log n)$ nodes in the tree. 
The assignment of a rectangle to a node is decided by the following three cases:

\vgap
\noindent
{\it Case-I.} The $xy$-projection of a rectangle  intersects none of the grid points, i.e., 
it lies completely inside one of the row slab or/and the column slab. 
Then the rectangle is not assigned to this node, but sent to the child node 
corresponding to the row or column  the rectangle lies in.

\vgap
\noindent
{\it Case-II.} The $xy$-projection of a rectangle $r$  intersects at least one of the grid points. 
Let $c_l$ and $c_r$ be the leftmost and the rightmost column of the grid intersected by $r$. 
Similarly, let $r_b$ and $r_t$ be the bottommost and the topmost row of the grid intersected by $r$.

Then the rectangle is broken into at most five disjoint pieces: 
 a {\em grid rectangle}, which  is the bounding box of all the grid points lying inside $r$ (see Figure~\ref{fig:type-II}(b)),
 two {\em column rectangles}, which  are the portions of $r$ lying in column $c_l$ and $c_r$ (see Figure~\ref{fig:type-II}(d)), and 
 two {\em row rectangles}, which are the remaining portion of the rectangle $r$ lying in row $r_b$ and $r_t$ (see Figure~\ref{fig:type-II}(c)). 
The grid rectangle is {\em assigned} to the node. Note that each column rectangle (resp., row rectangle) is now 
a $4$-sided rectangle in $\IR^3$ w.r.t. the column (resp., row) it lies in, and is sent 
to its corresponding child node.

\begin{figure}[h]
 \centering
\includegraphics[scale=0.4]{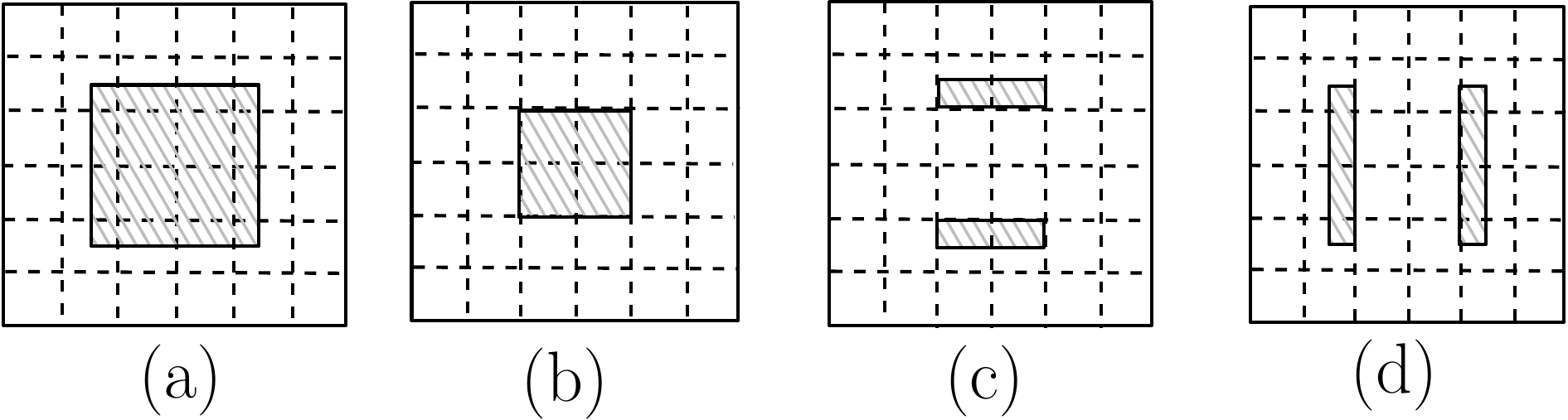}
\caption{}
\label{fig:type-II}
\end{figure}

\vgap
\noindent
{\it Case-III.} The $xy$-projection of a {\em $4$-sided rectangle} $r$ intersects at least one of 
the grid points. Without loss of generality, assume that the $4$-sided rectangle $r$ is 
unbounded along the negative $x$-axis. 
Then the rectangle is broken into at most four disjoint pieces: 
 a {\em grid rectangle,} as shown in Figure~\ref{fig:type-III}(b),
 one  {\em column rectangle}, as shown in   Figure~\ref{fig:type-III}(d), and 
 two {\em row rectangles}, as shown in   Figure~\ref{fig:type-III}(c). 
The grid rectangle and the two row rectangles are {\em assigned} to the 
node.
Note that the two row rectangles are now $3$-sided rectangles in $\IR^3$ w.r.t. their 
corresponding rows (unbounded in one direction along $x-$, $y-$ and $z-$axis). 
The column rectangle is sent to its corresponding child node. 
Analogous partition is performed for $4$-sided rectangles which are 
unbounded along positive $x$-axis, positive $y$-axis and negative 
$y$-axis.

\begin{figure}[h]
 \centering
\includegraphics[scale=0.4]{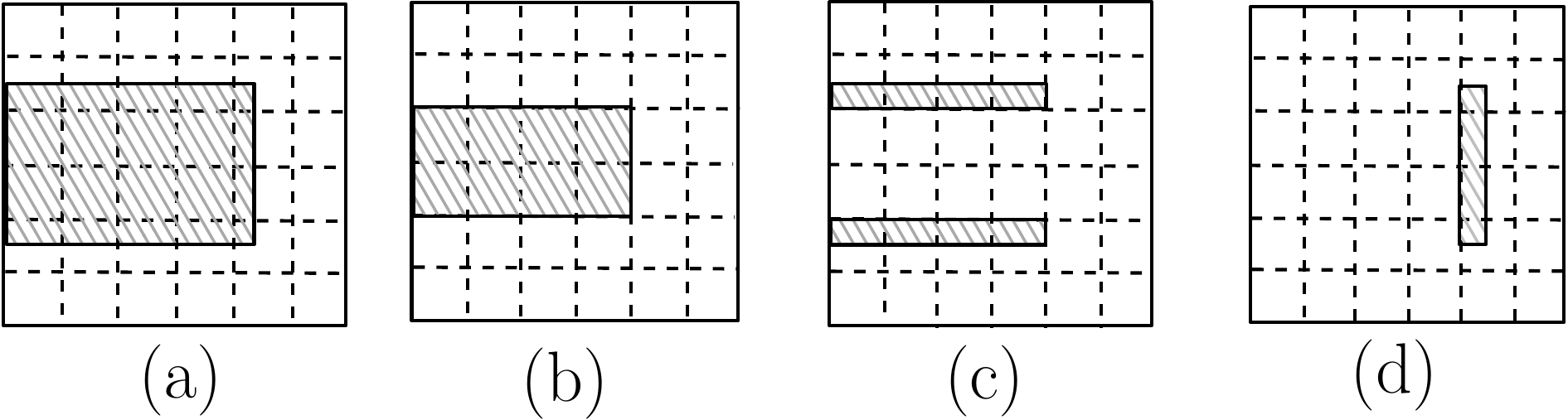}
\caption{}
\label{fig:type-III}
\end{figure}

\begin{observation}\label{obs:rec-tree}
A rectangle of $S$ gets assigned to at most four nodes at each level in the recursion tree.
\end{observation}

\begin{proof}
Consider a rectangle $r\in S$. 
If $r$ falls under Case-II, then its grid rectangle is assigned to the node. 
Note that $r$ can fall under Case-II only once, since each of its 
four row and column rectangles are now effectively $4$-sided rectangles.
Let $r'$ be one of these row or column rectangles.
If $r'$ falls under Case-III at a node, then it gets assigned there. 
However, this time exactly {\em one} of the broken portion of $r'$ will be sent to the child node.
Therefore, there can be at most four nodes at each level where rectangle $r$ (and broken portions of $r$) 
can get assigned.
\end{proof}

\vgap
\noindent
{\bf Data structures at each node.}  We build two types of structures 
at each node in the tree.

\vgap
\noindent
{\it Structure-I.}
A rectangle $r'$ is {\it higher} than rectangle $r''$ if $r'$ has a larger span than $r''$ along $z$-direction.
For each cell $c$ of the grid, based on the rectangles which completely cover $c$, 
we  construct a {\it sketch}  as follows: select the rectangle with the  $(1+\vare)^0, (1+\vare)^1,(1+\vare)^2,\ldots$-th 
largest span. For a given cell, the size of the sketch will be $O(\log_{1+\vare}m)$.

\vgap
\noindent
{\it Structure-II.} 
For a given row or column in the grid, let $\hat{S}$ be  the $3$-sided rectangles in $\IR^3$  
assigned to it. We build the linear-size structure of \cite{ahz10} on $\hat{S}$, 
which will return a $(1+\vare)$-approximation of  $|\hat{S}\cap q|$  in $\Oe(\log n)$ time. 
This structure is built for  each row and column slab.

\vspace{0.1 in}
\noindent
{\bf Space analysis.} Consider a node in the recursion tree with $m$ rectangles. 
There will be  $\left(2\sqrt{\frac{m}{t}}\right) \times \left(2\sqrt{\frac{m}{t}}\right)=4\frac{m}{t}$
cells at this node.
The space occupied by structure-I will be $O\left(\frac{m}{t}\cdot\log_{1+\vare}m \right)=O(m)$.
The space occupied by   structure-II  will be $O(m)$.  
Using Observation~\ref{obs:rec-tree}, 
 the total space occupied by all the nodes at a particular level will be $O(n)$.
Since the height of the recursion tree is $\Oe(\log\log n)$, the total space occupied is 
$\Oe(n\log\log n)$.

\vspace{0.1 in}
\noindent
{\bf Query algorithm.} Given a query point $q$, we start at the root node. At each visited node, 
the following three steps are performed:
\begin{enumerate}
\item  {\it Query structure-I.} Locate the cell $c$ on the grid containing $q$.
Scan the sketch of cell $c$ to return a $(1+\vare)$-approximation of 
the number of rectangles which cover $c$ and contain $q$. 
This takes $\Oe(\log m)$ time.

\item {\it Query structure-II.} Next, query structure-II of the horizontal and 
the vertical slab containing $q$, to find a $(1+\vare)$-approximation of the 
$3$-sided rectangles containing $q$.
 This takes $\Oe(\log m)$ time.
\item {\it Recurse.} Finally, we recurse on the horizontal and the vertical slab containing $q$. 
\end{enumerate}

The final output is the {\em sum} of the count returned by all the nodes queried.

\vspace{0.1 in}
\noindent
{\bf Query time analysis.}  Let $Q(n)$ denote the overall query time. 
Then 
\[Q(n) = 2Q(\sqrt{nt}) + \Oe(\log n), t=\log_{1+\vare} n.\]

This solves to $Q(n)=\Oe(\log n\cdot\log\log n)$. 
This finishes the proof of Lemma~\ref{lemma:subopt}.

\subsection{Final structure}

In this subsection we improve upon the data structure 
built in the previous subsection by reducing the size 
 to $\Oe(n\log^{*}n)$.

\begin{lemma}\label{lemma:almost-opt}
In the standard $5$-sided rectangle stabbing in $\IR^3$ problem, the input is a set $S$ 
 of $n$ $5$-sided rectangles in $\IR^3$  and the query $q$ is a point. 
Then there exists a pointer machine data structure of size $\Oe(n\log^{*}n)$ 
which can solve an approximate counting problem in $\Oe(\log n\cdot \log\log n)$ time.  
\end{lemma}

Let $\vare' \leftarrow \vare/4$ and $C>3$. The reason for choosing  these parameters 
will become clear later. We divide the solution into two cases. 

\subsubsection{Case-I: $k\in[0,C\vare'^{-2}\log n\cdot\log\log n]$ }
For the reporting version of $5$-sided rectangle stabbing in $\IR^3$ problem, 
Rahul~\cite{r15} presented a structure of size $O(n\log^{*}n)$ 
which can answer a query in $O(\log n\cdot \log\log n + k)$ time.
Build this structure on all the rectangles in set $S$. 
Given a query point $q$,  query the 
structure till all the rectangles in $S\cap q$ have been reported 
or $C\vare'^{-2}\log n\cdot\log\log n +1$ rectangles in $S\cap q$
have been reported. 
If the first event happens, then the exact value of $k$ is reported. 
Otherwise, we  conclude that $k >C\vare'^{-2}\log n\cdot\log\log n$.

\subsubsection{Case-II: $k\in[C\vare^{-2}\log n\cdot\log\log n,n]$}

We will need the following random sampling based lemma.

\begin{lemma}\label{lem:sampling}
Let $S$ be a set of $n$ $5$-sided rectangles in $\IR^3$. 
Consider a query point $q$ such that 
$k \geq C\vare'^{-2}\log n\cdot\log\log n$.
Then there exists a set $R \subset S$ of size $O\left(\frac{n}{\log\log n}\right)$ such 
that  $(|R\cap q|\cdot\log\log n) \in [(1-\vare')k, (1+\vare')k]$.
\end{lemma}

\begin{proof}
Fix a parameter $\delta=\log\log n$. 
Choose a random sample $R$ where each object of $S$ is picked independently with probability $1/\delta$. Therefore, the expected size of $R$ is $n/\delta$
(if the size of $R$ exceeds $O(n/\delta)$, then re-sample till we get the desired size). For a given 
query $q$,  $E[|R\cap q|]=|S\cap q|/\delta=k/\delta$.
Therefore, by Chernoff bound \cite{mp95} we observe that

\[ \textbf{Pr}\left[\left|\left|R\cap q\right| - \frac{k}{\delta}\right| > \vare' \frac{k}{\delta}\right] \leq e^{-\Omega(\vare'^2 (k/\delta))} \leq e^{-\Omega(\vare'^2(C\vare'^{-2}\log n))} \leq e^{-\Omega(C\log n)} = n^{-\Omega(C)} \leq o(1/n^C)\]

\noindent
Set $C$ to be greater than $3$. There are  $O(n^{3})$  combinatorially different query points on the 
set $S$. Therefore, by union bound it follows that there exists a subset $R\subset S$ of size $O(n/\delta)$ such that 
 $|k-|R\cap q|\cdot \delta| \leq \vare' k$, for any $q$ such that $k \geq C\vare'^{-2}\log n\cdot\log\log n$.
\end{proof}

\noindent
{\bf Preprocessing steps.} We perform the following steps:
\begin{itemize}
\item  Apply Lemma~\ref{lem:sampling} on set $S$ to obtain a 
set $R$ of size $O(n/\log\log n)$. 
\item  Build the data structure of Lemma~\ref{lemma:subopt} based on set $R$ with error parameter $\vare'$.
\end{itemize}

\noindent
{\bf Query algorithm.} 
For a given a query $q$, let $\tau_R$ be the value returned by  the data structure built on $R$. 
Then we report $\tau_R\cdot\log\log n$ as the answer. 

\vgap
\noindent
{\bf Analysis.}
Since $|R|=O(n/\log\log n)$, by Lemma~\ref{lemma:subopt} the space occupied by this data structure will be $\Oe(n)$. 
The query time follows from Lemma~\ref{lemma:subopt}. Next, we will prove that 
$(1-\vare)k \leq\tau_R\cdot\log\log n \leq (1+\vare)k$. 
  
\noindent
If we knew the exact value of $|R\cap q|$, then from Lemma~\ref{lem:sampling} we can infer that:
\begin{equation}
(1-\vare')k \leq |R\cap q|\log\log n \leq (1+ \vare')k 
\end{equation}
However, by using Lemma~\ref{lemma:subopt} we only get an approximate value of $|R\cap q|$:
\begin{equation}
(1-\vare')|R\cap q| \leq \tau_R \leq (1+ \vare')|R\cap q|
\end{equation}
Combining the above two equations, it is easy to verify that 
$(1-\vare)k \leq\tau_R\log\log n \leq (1+\vare)k$, where 
$\vare=4\vare'$.
This finishes the proof of Lemma~\ref{lemma:almost-opt}.

\begin{remark}\label{rem:ac-3d-dom}
\normalfont The general technique of Aronov and Har-Peled~\cite{ah08} 
can be adapted to answer the approximate counting query for the 
colored dominance search in $\IR^3$ problem. 
Assume that we have a data structure of size $S(n)$ which can 
answer the emptiness query in $Q(n)$ time. 
Ignoring the dependence on $\vare$, the technique of 
\cite{ah08} guarantees a data structure of size $O(S(n)\log^2n)$ 
which can answer a colored approximate counting query in 
$O(Q(n)\log n)$ time ($O(\log^2n)$ emptiness structures 
are built with each of them storing $\Theta(n)$ objects in the 
worst-case). For colored dominance 
search in $\IR^3$, plugging in $S(n)=O(n)$ and $Q(n)=O(\log n)$~\cite{p08b}, 
we get a data structure which requires  $\Oe(n\log^2n)$ space and 
$\Oe(\log^2n)$ query time.
\end{remark}

\section{Reduction-I: Reporting + $C$-approximation}\label{sec:first-red}
Our first reduction states that given a 
colored reporting structure and a colored $C$-approximation 
structure, one can obtain a colored $(1+\vare)$-approximation 
structure with no additional loss of efficiency. 
We need a few definitions before stating the theorem.
A geometric setting  is {\em polynomially bounded} if there are only 
$n^{O(1)}$ possible outcomes of $S\cap q$, over all possible 
values of $q$. For example, in $1d$ orthogonal range search on $n$ points, 
there are only $\Theta(n^2)$ possible outcomes of $S\cap q$.
A function $f(n)$ is {\it converging} if $\sum_{i=0}^t n_i =n, \text{ then } \sum_{i=0}^t f(n_i) = O(f(n))$.
For example, it is easy to verify that $f(n)=n\log n$ is converging.

\begin{theorem} \label{thm::main-1} 
For a colored geometric setting, assume that we are given the following two structures: 
\begin{itemize}
\item a colored reporting structure of $\S_\rep(n)$ size which can solve a query in 
$O(\Q_\rep(n) + \kappa)$ time, where $\kappa$ is the output-size, and 
\item a colored $C$-approximation structure of $\S_\capp(n)$ size which can solve 
a query in $O(\Q_\capp(n))$ time. 
\end{itemize}
We also assume that: (a) $\S_\rep(n)$ and $\S_\capp(n)$ are converging,  
and (b) the geometric setting is polynomially bounded. 
Then we can obtain a $(1+\vare)$-approximation using a structure that requires  $\S_\eapp(n)$ space and 
	$\Q_\eapp(n)$ query time, such that  
	\begin{eqnarray} 
	    \hspace{-10mm} &&\S_\eapp(n) = O(\S_\rep(n)+ \S_\capp(n)) \label{eqn::intro-ours1-space} \\
	    \hspace{-10mm} &&\Q_\eapp(n) = O\left(\Q_\rep(n) + \Q_\capp(n) + \vare^{-2}\cdot \log n \right). \label{eqn::intro-ours1-qry}
	\end{eqnarray}
\end{theorem}

\subsection{Refinement Structure}
The goal of  a refinement structure
is to convert  a constant-factor approximation of $k$ into a 
$(1+\vare)$-approximation of $k$.
\begin{lemma}\label{lem:refinement}
{\bf(Refinement structure)} Let ${\cal C}$ be the set of colors in set $S$, and 
 ${\cal C}\cap q$ be the set of colors in ${\cal C}$ present in $q$.
For a query $q$, assume we know that:
\begin{itemize}
\item $k=|{\cal C}\cap q|=\Omega(\vare^{-2}\log n)$, and 
\item $k\in [z,Cz]$, where $z$ is an integer.
\end{itemize}
Then there is a refinement structure of size $O\left(\S_\rep\left(\frac{\vare^{-2}n\log n}{z}\right)\right)$ which can report a value 
$\tau \in [(1-\vare)k,(1+\vare)k]$ in 
 $O(\Q_\rep(n) +\vare^{-2}\log n)$ time. 
\end{lemma}

The following lemma states that sampling colors (instead of input objects) is a 
useful approach to build the refinement structure. 

\begin{lemma}\label{lem:color-sampling-1}
Consider a query $q$ which satisfies the two conditions stated in Lemma~\ref{lem:refinement}. 
Let $c_1$ be a sufficiently large constant 
and $c$ be another constant s.t. $c=\Theta(c_1\log e)$.
Choose a random sample $R$ where each color in ${\cal C}$ is picked independently with 
probability $M= \frac{c_1\vare^{-2}\log n}{z}$. Then with probability $1-n^{-c}$ we have  
$\left|k-\frac{\left|R\cap q\right|}{M}\right| \leq \vare k$. 
\end{lemma}

\begin{proof}
For each of the $k$ colors which are present in $q$, 
define an indicator variable $X_i$. Set $X_i=1$, if the 
corresponding color is in the random sample $R$. 
Otherwise, set $X_i=0$. Then $|R \cap q|=\sum_{i=1}^k X_i$ and 
$E[|R \cap q|]=k\cdot M$. By Chernoff bound, 
\begin{align*}
\textbf{Pr}\Bigg[\Big||R \cap q| -E[|R\cap q|]\Big|>\vare\cdot E[|R \cap q|]\Bigg] 
<\text{exp}\Big(-\vare^{2}E[|R\cap q|]\Big) \\
<\text{exp}\left(-\vare^{2}\cdot kM\right) < \text{exp}\left(-\vare^{2}zM\right) 
<\text{exp}\left(-c_1\log n\right) \leq \frac{1}{n^{c}} 
\end{align*}
Therefore, with high probability $\Big||R \cap q| -kM\Big|\leq\vare\cdot kM$.
\end{proof}

\begin{lemma}\label{lemma:r-exists}
{\bf (Finding a suitable $R$)} Pick a random sample $R$ as defined in Lemma~\ref{lem:color-sampling-1}.
Let $n_R$ be the number of objects of $S$ whose color belongs to $R$.
We say $R$ is {\em suitable} if it satisfies the following two conditions:
\begin{itemize}
\item $\left|k-\frac{\left|R\cap q\right|}{M}\right| \leq \vare k$ for all  queries which have $k=\Omega(\vare^{-2}\log n)$.
\item $n_R \leq 10nM$. This condition is needed to bound the size of the data structure.
\end{itemize}
A suitable $R$ always exists.
\end{lemma}
\begin{proof}
Let $n^{\alpha}$ be the number of combinatorially different queries $q$ 
 on the set $S$. From Lemma~\ref{lem:color-sampling-1}, 
by setting $c=\alpha+1$, we can conclude that 
$\tau \longleftarrow \frac{\left|R\cap q\right|}{M}$ will lie in 
the range $[(1-\vare)k,(1+\vare)k]$  with 
probability at least $1-1/n^{\alpha+1}$. By the standard union bound, 
it implies that the probability 
of the  random sample $R$ failing for any query is at most $1/n^{\alpha+1}\times n^{\alpha}=1/n$.

Next, it is easy to observe that the {\em expected} value of $n_R$ is $nM$: 
Let $n_c$ be the number of objects of $S$ having color $c$. Then 
$\textbf{E}[n_R]=\sum_{\forall c}n_c\cdot M=nM$. 
By Markov's inequality, the probability of $n_R$ being larger than $10nM$ is 
less than or equal to $1/10$.  
By union bound, $R$ will be not be suitable with probability $\leq 1/n+1/10$. 
Therefore, with probability $\geq 9/10-1/n$, $R$ will be suitable and hence, we are done.
We do not discuss the preprocessing time here, since it is not known how to {\it efficiently}
verify if a sample $R$ is suitable.
We leave this as an interesting open problem. 
\end{proof}

\paragraph{Refinement structure and query algorithm.}
In the preprocessing stage  pick a random sample 
$R\subseteq {\cal C}$ as stated in Lemma~\ref{lem:color-sampling-1}.
If the sample $R$ is {\it not suitable}, then discard $R$ 
and re-sample, till we get a suitable sample. Based on all the 
objects of $S$ whose color belongs to $R$,  build a colored 
reporting structure. Given a query  $q$, the 
colored reporting structure is queried to compute $|R\cap q|$.
We report $\tau \longleftarrow \left(|R\cap q|/M\right)$ as the final answer. 
The query time is bounded by $O(\Q_\rep(n) + \vare^{-2}\log n)$, 
since by Lemma~\ref{lem:color-sampling-1}, $|R\cap q| \leq (1+\vare)\cdot kM =O(\vare^{-2}\log n)$.
This finishes the description of the refinement structure.

\subsection{Overall solution}

\paragraph{Data structure.} The data structure consists of the following three components:
\begin{enumerate}
\item {\em Reporting structure.} Based on the set $S$ we build  a colored reporting structure.
This occupies $O(\S_\rep(n))$ space.
\item {\em $\sqrt{C}$-approximation structure.} Based on the set $S$ we build a $\sqrt{C}$-approximation  
structure. The choice of $\sqrt{C}$ will become clear in the analysis. 
This occupies $O(\S_\capp(n))$ space.

\item {\em Refinement structures.} Build the
refinement structure of Lemma~\ref{lem:refinement} for the values 
$z=(\sqrt{C})^i\cdot \vare^{-2}\log n, \forall i\in \left[0, 
\log_{\sqrt{C}}\left(\left\lceil \vare^{2}n\right\rceil
\right)\right]$. The total size of all the 
refinement structures will be $\sum O\left(\S_\rep(nM)\right) = O(\S_\rep(n))$, since 
$\S_\rep(\cdot)$ is converging and $\sum nM=O(n)$. Note that our choice of 
$z$ ensures that the size of the data structure is independent of $\vare$.
\end{enumerate}

\paragraph{Query algorithm.} The query algorithm performs the following steps:
\begin{enumerate}
\item 
Given a query object $q$,  the colored reporting structure  reports the colors present in $S\cap q$ till 
all the colors have been reported or $\vare^{-2}\log n +1$ 
colors have been reported. If the first event happens, then the exact value of $k$ is reported.  
Otherwise, we  conclude that $k=\Omega(\vare^{-2}\log n)$.
This  takes $O(\Q_\rep(n) + \vare^{-2}\log n)$ time.
\item If $k > \vare^{-2}\log n$, then
\begin{enumerate}
\item First,  query the  $\sqrt{C}$-approximation structure. Let $k_a$ be the $\sqrt{C}$-approximate value returned s.t.
$k\in [k_a, \sqrt{C}k_a]$. This takes $O(\Q_\capp(n))$ time.
\item Then query the refinement structure with the largest value 
of $z$ s.t. $z\leq k_a  \leq \sqrt{C}z$. It is trivial to verify that 
$k  \in [z,Cz]$. This takes $O(\Q_\rep(n) +\vare^{-2}\log n)$ time.
\end{enumerate}
\end{enumerate}

 \section{Reduction-II: Using Only Reporting Structure}\label{sec:second-red}

In this section we will present our second general reduction. 
The reader is assumed to be familiar with Section~\ref{sec:first-red}.

\begin{theorem} \label{thm:accq} 

For a given colored geometric setting, assume that we are 
given  a colored reporting structure of $\S_\rep(n)$ size which can answer the query in 
$O(\Q_\rep(n) + \kappa)$ time. We also assume that: (a) $\S_\rep(n)$  is  converging,  
and (b) the geometric setting is {\em polynomially bounded}. 
Then we can obtain a $(1+\vare)$ approximation using a structure which requires  $\S_\eapp(n)=O(\S_\rep(n))$ space and 
	\begin{eqnarray} 
	    \label{eqn:insensitive}\hspace{-10mm} &&\Q_\eapp(n) = O\bigg(\big(\Q_\rep(n) +  \vare^{-2}\cdot \log n \big)
	    \cdot \log(\log_{1+\vare} |{\cal C}|)\bigg)
	\end{eqnarray}
\noindent
query time, where ${\cal C}$ is the number of colors in $S$.	
\end{theorem}

Similar to Section~\ref{sec:first-red}, a colored reporting structure will be built on $S$ 
to either report the exact value of  $|{\cal C}\cap q|$ or report that $|{\cal C}\cap q|$ is greater than 
$\vare^{-2}\log n $. From now on we will assume that $k=|{\cal C}\cap q|=\Omega(\vare^{-2}\log n)$.

\subsection{Decision structure}\label{decision}

\begin{lemma}\label{lem:decision}
{\bf(Decision structure)} Let $z=\Omega(\vare^{-2}\log n)$ be a pre-specified parameter.
Given a query $q$, the decision structure reports whether 
$|{\cal C}\cap q| \geq z$ or $|{\cal C}\cap q| < z$. 
The data structure is allowed to make a mistake when 
$|{\cal C} \cap q| \in [(1-\vare)z,(1+\vare)z]$.
There is a decision structure of size $O\left(\S_\rep\left(\frac{\vare^{-2}n\log n}{z}\right)\right)$ 
which can answer the query in $O(\Q_{rep}(n) +\vare^{-2}\log n)$ time.
\end{lemma}

In this subsection we will prove the above lemma. 
A few words on the intuition behind the solution. Suppose each color in ${\cal C}$ is sampled 
with probability $\approx (\log n)/z$. For a given query $q$, if $k < z$ (resp., $k > z$), then the expected number of 
colors from ${\cal C} \cap q$ sampled will be less than $\log n$ (resp., greater than $\log n$). 
 We will start by proving the following lemma.

\begin{lemma}\label{lem:failure}
Let $c_1$ be a sufficiently large constant  
and $c$ be another constant s.t. $c=\Theta(c_1\log e)$.
Consider a random sample $R$ where each color in ${\cal C}$ is picked independently with 
probability $M= \frac{c_1\vare^{-2}\log n}{z}$, where $\vare \in (0,1/2]$. Then 
 \[\textbf{Pr}\bigg[|R\cap q| > zM \hspace{2 mm} \bigg|\hspace{2 mm} k \leq (1-\vare)z\bigg] \leq \frac{1}{n^{c}}.\]
Similarly, 
\[\textbf{Pr}\bigg[|R\cap q| \leq zM \hspace{2 mm}\bigg|\hspace{2 mm} k \geq (1+\vare)z\bigg] \leq \frac{1}{n^{c}}\]  
\end{lemma}

\begin{proof}
For each of the $k$ colors present in $q$, 
define an indicator variable $X_i$. Set $X_i=1$ if the 
corresponding color is in the random sample $R$.
Otherwise, set $X_i=0$. Then $|R \cap q|=\sum_{i=1}^k X_i$ and 
$E[|R \cap q|]=k\cdot M$. For the sake of brevity, let 
$Y=|R\cap q|$. We only prove the first fact here. The proof for the second fact is similar. Let

\vgap
\noindent
 \[\alpha = \textbf{Pr}\bigg[Y > zM \hspace{2 mm} \bigg|\hspace{2 mm} k \leq (1-\vare)z\bigg]  \]
 The value $\alpha$ is maximized when $k=(1-\vare)z$. Therefore, 
 \[\alpha \leq \textbf{Pr}\bigg[Y > zM \hspace{2 mm} \bigg|\hspace{2 mm} k = (1-\vare)z\bigg] \]
 In this case, $\textbf{E}[Y]=kM=(1-\vare)zM$. Therefore,
\begin{align*}
\alpha&\leq \textbf{Pr}[Y>zM] = \textbf{Pr}\left[Y> \frac{1}{1-\vare}\textbf{E}[Y]\right] \leq \textbf{Pr}\left[Y> (1+\vare)\textbf{E}\left[Y\right]\right] \\
&\leq \text{exp}\left(-\frac{\vare^2\textbf{E}[Y]}{4}\right) \quad \quad \text{By Chernoff bound}\\
&=\text{exp}\left(-\vare^2(1-\vare)z\left(\frac{c_1\vare^{-2}\log n}{4z}\right)\right) =\text{exp}\left(-c_1(1-\vare)\frac{\log n}{4}\right)
\leq \text{exp}\left(-\frac{c_1}{8}\log n\right)  \quad \text{ since } \vare \leq 1/2\\
&\leq \frac{1}{n^{c}}
 \end{align*}
\noindent

\end{proof}

\begin{lemma}
 Let $z=\Omega(\vare^{-2}\log n)$ be a pre-specified parameter.
Using notation from Section~\ref{sec:first-red}, 
a sample $R \subseteq {\cal C}$ is called {\it suitable} if 
\begin{itemize}
\item For all queries, (a) if $k < (1-\vare)z$ then $|R\cap q| < c_1\vare^{-2}\log n$, 
and (b) if $k \geq (1+\vare)z$ then $|R\cap q| \geq c_1\vare^{-2}\log n$.
\item $n_R \leq 10nM$.
\end{itemize}

Such an $R$ always exists.
\end{lemma}

\begin{proof}
The proof is exactly the same as the proof in Lemma~\ref{lemma:r-exists}. 
The only difference is that we replace Lemma~\ref{lem:color-sampling-1} 
with Lemma~\ref{lem:failure}. 
\end{proof}

\paragraph{Decision structure and query algorithm.}
In the preprocessing phase  pick a random sample 
$R\subseteq {\cal C}$ as stated in Lemma~\ref{lem:failure}.
If the sample $R$ is {\it not suitable}, then  discard $R$ 
and re-sample, till we get a suitable sample. Based on all the 
points of $S$ whose color belongs to $R$,  build a colored 
reporting structure. Given a query object $q$,  the 
colored reporting structure reports  $R\cap q$, 
till  all the colors  
 have been reported or $c_1\vare^{-2}\log n$ colors
have been reported.
If the first event happens, then we report $k < z$. Otherwise,
 we report $k\geq z$. The query time is bounded by $O(\Q_\rep(n) + \vare^{-2}\log n)$, 

In Lemma~\ref{lem:failure}, we assumed $\vare \in (0,1/2]$. 
Handling  $\vare \in (1/2,1]$ is easy: Set a new variable 
$\vare_{new} \longleftarrow 1/2$. The decision structure will be built with the error  parameter $\vare_{new}$ (and not $\vare$). 
Since $\vare_{new} < \vare$,  the error made by the decision structure is tolerable.
Since $\frac{1}{\vare_{new}} \leq \frac{2}{\vare}$, the space and the query time bounds are also not affected.
\subsection{Final structure}

\noindent
{\bf Data structure.}  Recall that we only have to handle $k=\Omega(\vare^{-2}\log n)$.
For the values $z_i=c_1(\vare^{-2}\log n)(1+\vare)^i$, for $i=1, 2,3,\ldots, W=O(\log_{1+\vare} |{\cal C}|)$, 
we build a decision structure ${\cal D}_i$  using Lemma~\ref{lem:decision}. 
 By performing similar analysis as in Section~\ref{sec:first-red}, 
the overall size will be  $O(\S_\rep(n))$.

\vgap
\noindent
{\bf Query algorithm.}
For a moment, assume that we query all the data structures ${\cal D}_1,\ldots,{\cal D}_W$. 
Then  we will see a sequence of structures ${\cal D}_j$ for $j\in [1,i]$ 
claiming $|{\cal C}\cap q| > z_j$, followed by a sequence 
of structures ${\cal D}_{i+1},\ldots, {\cal D}_W$ claiming  $|{\cal C}\cap q| \leq z_j$. 
Then we report $\tau\leftarrow z_i$ as the answer 
to the approximate colored counting query.  
A simple calculation reveals that $\tau$ will lie in the range $[(1-\vare)k,(1+\vare)k]$. 
We perform a binary search on ${\cal D}_1,\ldots,{\cal D}_W$ to efficiently find the index $i$. 
The query time will be $O\bigg(\big(\Q_\rep(n) +  \vare^{-2}\cdot \log n \big)
	    \cdot \log(\log_{1+\vare} |{\cal C}|)\bigg)$. 

\paragraph{Remark.} Our result is a generalization of the reduction of Aronov and Har-Peled~\cite{ah08} 
to colored problems. 
Handling ``small" values of $k$ efficiently is usually challenging, since the error tolerated is small. 
Using the reporting structure makes it easy to handle the ``small" values of $k$  
(unlike an emptiness structure which was used by \cite{ah08}). 
Random sampling and Chernoff bound are easy to apply for ``large" values of $k$. 
As a result, the analysis of our reduction is easier than \cite{ah08}.

\section{Colored Orthogonal Range Search in $\IR^2$}\label{sec:three-sided-color}
To illustrate an application of Reduction-I, we study the approximate  colored counting query 
for orthogonal range search in $\IR^2$. 

\begin{theorem} \label{thm::three-sided-color} 
Consider the following two problems:
\begin{enumerate}[label=\Alph*)]
 \item {\bf Colored $3$-sided range search in $\IR^2$.} 
In this setting, the input set $S$ is $n$ colored points in $\IR^2$ and the query $q$ 
is a $3$-sided rectangle in $\IR^2$. 
There is a data structure 
of $O(n)$ size which can answer the approximate colored counting query in $O(\vare^{-2}\log n)$  time. 
This pointer machine structure is optimal in terms of $n$.

\item {\bf Colored $4$-sided range search in $\IR^2$.} 
In this setting, the input set $S$ is $n$ colored points in $\IR^2$ and the query $q$ 
is a $4$-sided rectangle in $\IR^2$.  
There is a data structure 
of $O(n\log n)$ size which can answer the approximate colored counting query in $O(\vare^{-2}\log n)$  time. 
\end{enumerate}
\end{theorem}

\subsection{Colored $3$-sided range search in $\IR^2$}

We use the framework of Theorem~\ref{thm::main-1} to 
prove the result of Theorem~\ref{thm::three-sided-color}(A). 
For this geometric setting, a colored reporting structure with $\S_\rep =n$ and $\Q_\rep=\log n$ is already known \cite{sj05}. 
The path-range tree of Nekrich~\cite{n14} gives a $(4+\vare)$-approximation, 
but it requires  super-linear space. The $C$-approximation structure presented in this subsection is
 a refinement of the path-range tree for the pointer machine model.

\begin{lemma}\label{lem:3-sided-c}
For the colored $3$-sided range search in $\IR^2$ problem, there is a 
$C$-approximation structure which requires $O(n)$ space and 
answers a query in $O(\log n)$ time.
\end{lemma}

\noindent
We prove Lemma~\ref{lem:3-sided-c} in  the rest of this subsection.

\subsubsection{Interval tree}
Our solution is based on an interval tree and  we will need the following fact about it.
\begin{lemma}\label{lem:it}
Using interval trees, a query on $(3+t)$-sided rectangles in $\IR^3$ can be 
broken down into $O(\log n)$ queries on $(2+t)$-sided rectangles in $\IR^3$. 
Here $t\in [1,3]$.
\end{lemma}

\begin{proof}
Let $R$ be a set of $n$ $(3+t)$-sided rectangles.
We build an interval tree ${\cal IT}$ as follows: W.l.o.g., assume that the rectangles are 
bounded along the $x$-axis.  Let $h$ be a plane perpendicular to the $x$-axis such that 
there are equal number of endpoints of $R$ on each side of the plane. 
The  splitting halfplane $h$ is stored at the root of ${\cal IT}$ and  the two subtrees are 
built recursively. In general,  $h(v)$ is the splitting halfplane stored at a node $v \in {\cal IT}$. 
 A rectangle $r\in R$ is stored at the highest node $v$ s.t. $r$ intersects 
$h(v)$. Let $R_v$ be the set of rectangles stored at a node $v$.
Each rectangle in $r\in R_v$ is split by $h(v)$ into two rectangles $r^-$ and $r^+$. 
Define $R_v^- := \bigcup_{r\in R_v} r^-$ and $R_v^+:=\bigcup_{r\in R_v} r^+$. 

Given a query point $q$,  trace a path $\Pi$ of length $O(\log n)$ from the root to a leaf node corresponding 
to $q$. For a node $v\in \Pi$,  if  $q$ lies to the left (resp., right) of $h(v)$, then 
answering a query on $R_v \cap q$ is equivalent to answering it on $R_v^- \cap q$ (resp., $R_v^+ \cap q$), 
and  we can treat $R_v^-$ (resp., $R_v^+$) as $(2+t)$-sided rectangles in $\IR^3$, since 
$h(v)$ is effectively $+\infty$ (resp., $-\infty$). 
\end{proof}

\subsubsection{Initial structure} 

\begin{lemma}\label{lem:is}
For the colored $3$-sided range search in $\IR^2$ problem, there is a 
$2$-approximation structure which requires $O(n)$ space and 
answers a query in $O(\log^3 n)$ time.
\end{lemma}
\begin{proof}
By a simple exercise, the colored $3$-sided range search in $\IR^2$ can 
be reduced to the colored dominance search in $\IR^3$. Therefore, using the 
reduction of Subsection~\ref{subsec:red-colored-2} the 
colored $3$-sided range search in $\IR^2$ also reduces to 
standard $5$-sided rectangle stabbing problem (for brevity, call it $5$-sided RSP).

There is a simple linear-size data structure which reports in $O(\log^3 n)$ time 
a $2$-approximation for the $5$-sided RSP: By inductively applying Lemma~\ref{lem:it} twice, 
we can decompose  $5$-sided RSP 
 to $O(\log^2n)$ $3$-sided RSPs. For $3$-sided RSP, there is a linear-size structure of 
 which reports a $2$-approximation  
 in $O(\log n)$ time~\cite{ahz10}.  By using this structure the 
 $5$-sided RSP can be solved in   $O(\log^3n)$ time. 
\end{proof}

\subsubsection{Final structure} 
Now we will present the optimal $C$-approximation structure of Lemma~\ref{lem:3-sided-c}.

{\em Structure.} Sort the points of $S$ based on their $x$-coordinate value and divide them into buckets containing 
$\log^2n$ consecutive points. Based on the points in each bucket, build a $D$-structure which is an instance of 
Lemma~\ref{lem:is}. Next, build a height-balanced binary search tree ${\cal T}$, 
where the buckets are placed at the leaves from left to right based on their ordering along the $x$-axis. 
Let $v$ be a proper ancestor of a leaf node $u$ and let $\Pi(u,v)$ be the path from $u$ to $v$ 
(excluding $u$ and $v$). Let $S_l(u,v)$ be the set of points in the subtrees rooted at nodes 
that are left children of nodes on the path $\Pi(u,v)$ but not themselves 
on the path.  Similarly, let $S_r(u,v)$ be the set of points in the subtrees rooted at nodes 
that are right children of nodes on the path $\Pi(u,v)$ but not themselves 
on the path. See Figure~\ref{fig:color-3-sided}, which illustrates these sets for two leaves 
$u=u_l$ and $u=u_r$.  
For each pair $(u,v)$, let $S_l'(u,v)$ (resp., $S_r'(u,v)$) be the set of points that each have the highest 
$y$-coordinate value among the points of the same color in $S_l(u,v)$ (resp., $S_r(u,v)$). 

Finally, for each pair $(u,v)$, construct a {\it sketch}, $S_l''(u,v)$, 
 by selecting the  $2^0, 2^1,2^2,\ldots$-th highest $y$-coordinate point 
in $S_l'(u,v)$.
A symmetric construction is performed to obtain $S_r''(u,v)$. 
The number of $(u,v)$ pairs is bounded by  $O((n/\log^2n)\times(\log n))=O(n/\log n)$ and 
hence, the space occupied by all the $S_l''(u,v)$ and 
$S_r''(u,v)$ sets is $O(n)$. 

\begin{figure}[h]
 \centering
\includegraphics[scale=1]{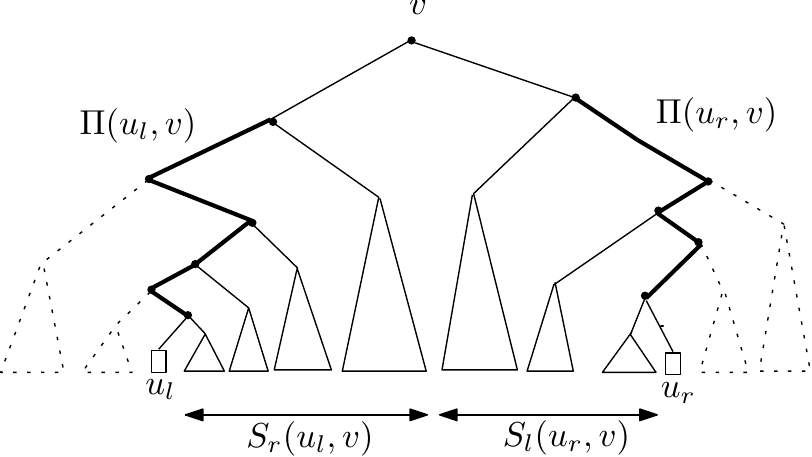}
\caption{}
\label{fig:color-3-sided}
\end{figure} 

{\em Query algorithm.} To answer a query $q=[x_1,x_2] \times [y,\infty)$, we first determine the leaf nodes 
$u_l$ and $u_r$ of ${\cal T}$ containing $x_1$ and $x_2$, respectively. 
If $u_l =u_r$, then we query the $D$-structure corresponding to the leaf node and 
we are done. If $u_l \neq u_r$, then we find the node $v$ which is the 
least common ancestor of $u_l$ and $u_r$. The query is now broken into four sub-queries:
First, report the approximate count in the leaves $u_l$ and $u_r$  by querying the $D$-structure of $u_l$ with 
$[x_1,\infty) \times [y,\infty)$ and the $D$-structure of $u_r$ with $(-\infty, x_2] \times [y,\infty)$. 
Next,  scan the list $S_r''(u_l,v)$ (resp., $S_l''(u_r,v)$) 
to find a $2$-approximation of the number of colors of $S_r(u_l,v)$ 
(resp., $S_l(u_r,v)$) present in $q$.
 
 The final answer is the sum of the count returned by the four sub-queries. 
 The time taken to find $u_l$, $u_r$ and $v$ is  $O(\log n)$. 
 Querying the leaf structures takes
 $O((\log (\log^2n))^3)=O(\log n)$ 
 time. The time taken for scanning the lists $S_r''(u_l,v)$ and $S_l''(u_r,v)$ 
 is  $O(\log n)$. Therefore, the overall query time is bounded by $O(\log n)$.
 Since each of the four sub-queries give a $2$-approximation, overall we get a
 $8$-approximation. 
 
\subsection{$C$-approximation for $4$-sided range search}
 
 Now we will prove Theorem~\ref{thm::three-sided-color}(B). 
 Again we will use the framework of Theorem~\ref{thm::main-1}.
  It is straightforward to obtain 
a data structure with $\S_{\capp}=O(n\log n)$, $\Q_{\capp}=O(\log n)$ and 
$C=16$. Simply build a binary range tree on the $y$-coordinates of $S$ and at each node 
build an instance of Lemma~\ref{lem:3-sided-c} based on the points 
in its subtree. Given a $4$-sided query rectangle $q$, it can be broken 
down into two $3$-sided query rectangles. Shi and Jaja \cite{sj05} 
presented a reporting structure with $\S_{\rep}=O(n\log n)$ 
and $\Q_{\rep}=O(\log n)$. Plugging in these values into 
Theorem~\ref{thm::main-1} proves Theorem~\ref{thm::three-sided-color}(B).

\begin{remark}\label{rem:ac-3-sided}
\normalfont As discussed in Remark~\ref{rem:ac-3d-dom}, the technique of 
\cite{ah08} can be adapted to answer a colored approximate counting query. 
For colored $3$-sided range search in $\IR^2$, plugging in $S(n)=O(n)$ 
and $Q(n)=O(\log n)$~\cite{m85} leads to a data structure of size $O(n\log^2n)$ 
and query time $O(\log^2n)$. 
For colored $4$-sided range search in $\IR^2$, plugging in $S(n)=O(n\log n)$ 
and $Q(n)=O(\log n)$~\cite{bcko08} leads to a data structure of size $O(n\log^3n)$ 
and query time $O(\log^2n)$ (the structure of Chazelle~\cite{c86} can be used 
to obtain slightly better space). 
\end{remark}

\section{Applications of Reduction-II}\label{sec:appl-sec-red}

In this section we present a few applications of reduction-II. 
For the colored problems discussed in this section, their exact counting structures are expensive~\cite{gjs04,krsv07}. 
\begin{theorem} \label{thm::halfspace-color} 
Consider the following three colored geometric settings:
\begin{enumerate}
\item {\bf Colored halfplane range search}, 
where the input is a set of $n$ colored points in $\IR^2$ 
and the query is a halfplane. 
There is a data structure 
of $O(n)$  size which can answer the approximate counting query in 
$O\left(\frac{1}{\vare^2}\cdot\log n\cdot\log(\log_{1+\vare} |{\cal C}|)\right)$   
  time.

\item {\bf Colored halfspace range search  in $\IR^3$}, 
where the input is a set of $n$ points in $\IR^3$ and the query is a halfspace. 
There is a data structure 
of $O(n\log n)$  size which can answer the approximate counting query in 
$O\left(\frac{1}{\vare^{2}}\log^3n\cdot \log(\log_{1+\vare} |{\cal C}|)\right)$
 time. 

\item  {\bf Colored orthogonal range search  in $\IR^d$}, 
where the input is a set of $n$ points in $\IR^d$ and the query is an axis-parallel 
rectangle. There is a data structure 
of $O(n\log^d n)$  size which can answer the approximate counting query in 
$O\left(\frac{1}{\vare^{2}}\cdot\log^{d+1}n \cdot\log(\log_{1+\vare} |{\cal C}|) \right)$
  time. 
\end{enumerate}
\end{theorem}

\vgap
\noindent
{\bf Colored orthogonal range search in $\IR^d$.} 
First consider the {\it standard} orthogonal range emptiness query in 
$\IR^d$ ($d\geq 2$).  Using range trees, 
this problem can be solved using $M(n)=O(n\log^{d-1}n)$ space and 
$\Q_{\rep}(n)=O(\log^{d-1}n)$ query time. Using this structure,
the colored orthogonal range reporting problem in $\IR^d$ can be answered 
in $O(\Q_{\rep}(n) + \kappa \Q_{\rep}(n)\log n)$ query time 
using a structure of size $O(M(n)\log n)$. Here $\kappa$ is the 
number of colors reported (see Section~$1.3.4$ of \cite{gjs04} for the 
details of this transformation). 

By applying Theorem~\ref{thm:accq}, the space occupied by the
 approximate counting structure will be 
$O(n\log^{d}n)$.
In Theorem~\ref{thm:accq}, we assumed that the query time of the colored 
reporting structure can be expressed as $O(\Q_{\rep} + \kappa)$, whereas 
for this problem the query time is being expressed as $O(\Q_{\rep} + \kappa \Q_{\rep}\log n)$. 
Therefore, equation~\ref{eqn:insensitive} of the query time $\Q_{\eapp}(n)$ in Theorem~\ref{thm:accq} can be rewritten as 

\[ O\bigg(\big(\Q_{\rep}(n) + (\vare^{-2}\log n) \Q_{\rep}(n)\log n \big)\cdot\log(\log_{1+\vare} |{\cal C}|)\bigg)\]

Plugging in the value of $\Q_{\rep}(n)$  into the above expression, we get  
\[\Q_{\eapp}(n) =O\bigg(\big(\log^{d-1}n + \vare^{-2}\log^{d+1} n\big)\cdot\log(\log_{1+\vare} |{\cal C}|)\bigg) 
= O\bigg( \vare^{-2}\cdot\log^{d+1}n \cdot\log(\log_{1+\vare} |{\cal C}|) \bigg) \]

\vgap
\noindent
{\bf Colored halfspace range search in $\IR^3$.}  There exists an 
$O(n\log^2n)$ space 
reporting data structure for this problem which can answer the query in 
$O(n^{1/2+\delta}+ \kappa)$ time \cite{gjs04}. But we will not use this structure, 
since for our purpose $\kappa=O(\vare^{-2}\log n)$ and the transformation technique 
used for the colored orthogonal range search problem 
will give us a reporting data structure with better bounds. 
Again, first consider the {\it standard} halfspace range emptiness query in 
$\IR^3$. This problem can be solved using $M(n)=O(n)$ space and 
$\Q_{\rep}(n)=O(\log n)$ query time \cite{ac09b}. Using this structure,
the colored halfspace range reporting problem in $\IR^3$ can be answered 
in $O(\Q_{\rep}(n) + \kappa \Q_{\rep}(n)\log n)=O(\vare^{-2}\log^3n)$ query time 
using a structure of size $O(M(n)\log n)=O(n\log n)$.

Applying Theorem~\ref{thm:accq}, the space occupied by the
 approximate counting structure will be 
$O(n\log n)$. The query time will be
\[ O\bigg(\big(\Q_{\rep}(n) + (\vare^{-2}\log n) \Q_{\rep}(n)\log n \big)\cdot\log(\log_{1+\vare} |{\cal C}|)\bigg)
= O\bigg( (\vare^{-2}\log^3n)\cdot \log(\log_{1+\vare} |{\cal C}|)\bigg)\]

\vgap
\noindent
{\bf Colored halfplane range search.}  In \cite{gjs04} 
a reduction is presented from this problem to the {\it segments-below-point} problem: 
Given a set of $n$ segments in the plane, report all the segments
hit by a vertical query ray. Recently, this segments-below-point has been solved 
by Agarwal, Cheng, Tao and Yi \cite{acty09} in the context of designing data structures 
for uncertain data. They present an $O(n)$ space data structure to solve the query in 
$O(\log n + \kappa)$ time. This implies a solution for colored halfplane range reporting 
 with the same bounds. Plugging in this result into Theorem~\ref{thm:accq}, we 
obtain an $O(n)$ size data structure which can answer the approximate counting query in 
$O(\vare^{-2}\log n\cdot \log(\log_{1+\vare} |{\cal C}|))$ time.

\section{Conclusion and Future Work}
In this work, we built optimal and near-optimal approximate counting data structures for several 
colored and uncolored (i.e., standard) geometric settings.
We finish by presenting a few open problems:
\begin{enumerate}
\item We do not discuss the  preprocessing time in this paper since most of the solutions 
in this work are based on verifying if a sample is ``suitable". It is not known how 
to efficiently verify if a sample  is suitable. 
\item In a colored orthogonal range search in $\IR^1$ problem, the input is a set of 
$n$ colored points on the real-line and the query is an interval. Is it possible to 
build a linear-space data structure which can answer the approximate counting query 
in $\Oe(1)$ time in the word-RAM model?  
\end{enumerate}

\vgap
\noindent
{\bf Acknowledgements.} I am thankful to Prof. Sariel Har-Peled, Prof. Ravi Janardan, Yuan Li,
Sivaramakrishnan Ramamoorthy, Stavros Sintos, Jie Xue for fruitful 
discussions on these problems and for comments on the previous drafts. I would also like to 
thank the anonymous referees whose comments helped in immensely improving 
the content and the presentation of the paper.

\bibliographystyle{plain}
\bibliography{./ref}

\end{document}